\documentclass[a4paper,UKenglish,cleveref, autoref, thm-restate]{lipics-v2021}

\DeclareMathOperator{\overlap}{overlap}

\usepackage{tikz}
\usetikzlibrary{calc,arrows,shapes,backgrounds,patterns,fit,decorations,decorations.pathmorphing,math}

\tikzstyle{hgedge}=[->,gray!40!white]
\tikzstyle{anypath}=[->,dashed]
\tikzstyle{vertex}=[draw,ellipse,inner sep=0.5mm,minimum size=4mm]
\tikzstyle{inputvertex}=[draw,rectangle,inner sep=.5mm,minimum size=5mm]

\newenvironment{mypic}{\begin{center}\begin{tikzpicture}[>=latex,line width=.3mm,scale=0.8,transform shape]}{\end{tikzpicture}\end{center}}

\title{All instantiations of the greedy algorithm for the shortest superstring problem are equivalent} 

\titlerunning{All instantiations of the greedy algorithm for SCS are equivalent} 

\author{Maksim S. Nikolaev}{Steklov Institute of Mathematics at St.~Petersburg, Russian Academy of Sciences}{makc-nicko@yandex.ru}{https://orcid.org/0000-0003-4079-2885}{}

\authorrunning{M.\,S. Nikolaev} 

\Copyright{Maksim S. Nikolaev} 



\category{} 

\relatedversion{} 

\nolinenumbers 

\hideLIPIcs  

\EventEditors{John Q. Open and Joan R. Access}
\EventNoEds{2}
\EventLongTitle{32nd Annual Symposium on Combinatorial Pattern Matching (CPM 2021)}
\EventShortTitle{CPM 2021}
\EventAcronym{CPM}
\EventYear{2020}
\EventDate{July 5--7, 2021}
\EventLocation{Wroclaw, Poland}
\EventLogo{}
\SeriesVolume{42}
\ArticleNo{23}

\begin{document}
    
    \maketitle
    
    \begin{abstract}
        In the Shortest Common Superstring problem (SCS), one needs to find the shortest superstring
        for a set of strings. While SCS is NP-hard and MAX-SNP-hard, the Greedy Algorithm ``choose
        two strings with the largest overlap; merge them; repeat'' achieves a constant factor approximation
        that is known to be at most 3.5 and conjectured to be equal to 2. The Greedy Algorithm is not
        deterministic, so its instantiations with different tie-breaking rules may have different approximation
        factors. In this paper, we show that it is not the case: all factors are equal. To prove this, we show
        how to transform a set of strings so that all overlaps are different whereas their ratios stay roughly
        the same.
        
        We also reveal connections between the original version of SCS and the following one: find a~superstring minimizing the number of occurrences of a given symbol. It~turns out that the latter problem is equivalent to the original one.
    \end{abstract}
    
    \section{Introduction}
    \label{sec:intro}
        In the Shortest Common Superstring problem (SCS), one is given a set of strings and needs to find the shortest string that contains each of them as a~substring. Applications of this problem include genome assembly~\cite{waterman1995introduction, pevzner2001eulerian} and data compression~\cite{GMS1980, phdthesis, storer1987data}. We refer the reader to the surveys~\cite{gevezes2014shortest, mucha2007tutorial} for an overview of~SCS as well as its applications and algorithms.
        
        While SCS is NP-hard~\cite{GMS1980} and even MAX-SNP-hard~\cite{BJLTY1991}, the Greedy Algorithm (GA) ``choose two strings with the largest overlap; merge them; repeat'' achieves a constant factor approximation that is proven to be less than or equal to 3.5~\cite{KS2005}. This factor is at least 2 (consider a dataset $\mathcal{S}= \{ c(ab)^n, (ab)^nc, (ba)^n \}$) and the $30$ years old \emph{Greedy Conjecture}~\cite{storer1987data, TU1988, T1989, BJLTY1991} claims that this bound is accurate, that is, that GA is 2-approximate.
        
        GA is not deterministic as we do not specify how to break ties in case when there are many pairs of strings with maximum overlap. For this reason, different instantiations of GA may produce different superstrings for the same input and hence they may have different approximation factors. In fact, if $\mathcal{S}$ contains only strings of length 2 or less or if $\mathcal{S}$ is a set of $k$-substrings of an unknown string, then there are instantiations of GA~\cite{GMKLN2019}, that find \emph{the exact} solution, whereas in general GA fails to do so.
        
        The original Greedy Conjecture states that \emph{any} instantiation of GA is 2-approximate. As this is still widely open, it is natural to try to prove the conjecture at least for \emph{some} instantiations. This could potentially be easier not just because this is a weaker statement, but also because a particular instantiation of GA may decide how to break ties by asking an~almighty oracle. In this paper, we show that this weak form of Greedy Conjecture is in fact equivalent to the original one. More precisely, we show, that if \emph{some} instantiation of GA is~$\lambda$-approximate, then \emph{all} instantiations are $\lambda$-approximate.
        
        To prove this, we introduce the so-called \emph{Disturbing Procedure}, that, for a given dataset $\mathcal{S} = \{s_1, \dots, s_n\}$, a parameter $m \gg n$, and a sequence of greedy non-trivial merges (merges of strings with a non-empty overlap), constructs a new dataset $\mathcal{S}' = \{s'_1, \dots, s'_n\}$, such that, for all $i\neq j$, $s_i'$ is roughly $m$ times longer than $s_i$, the overlap of $s'_i$ and $s'_j$ is roughly $m$ times longer than the overlap of $s_i$ and $s_j$, and the mentioned greedy sequence of non-trivial merges for $\mathcal{S}$ is \emph{the only} such sequence for $\mathcal{S}'$.
        
        We also find the following curious relation between SCS and its version, where one needs
        to find a superstring with the smallest number of occurrences of a given symbol: if there is a~$\lambda$-approximate algorithm for the latter problem, then there is a~$\lambda$-approximate algorithm
        for the former one, and vice versa.
    
    \section{Preliminaries}
    
        Let $|s|$ be the length of a string $s$ and $\overlap(s, t)$ be the \emph{overlap} of strings $s$ and $t$, that is, the longest string $y$, such that $s = xy$ and $t = yz$. In this notation, a string $xyz$ is \emph{a merge} of strings $s$ and $t$. By $\varepsilon$ we denote the empty string. By $\mathrm{OPT}(\mathcal{S})$ we denote the optimal superstring for the dataset $\mathcal{S}$.
        
        Without loss of generality we may assume that the set of input strings $\mathcal{S}$ contains no~string that is a substring of another. This assumption implies that in any superstring all strings occur in some order: if one string begins before another, then it also ends before. Hence, we can consider only superstrings that can be obtained from some permutation $(s_{\sigma(1)}, \dots, s_{\sigma(n)})$ of $\mathcal{S}$ after merging adjacent strings. The length of such superstring $s(\sigma)$ is~simply
        \begin{gather}
            \label{length}
            |s(\sigma)| = \sum_{i=1}^n |s_i| - \sum_{i=1}^{n-1}|\overlap(s_{\sigma(i)}, s_{\sigma(i+1)})|.
        \end{gather}
        
        Let $A$ be an instantiation of GA (we denote this by $A \in \mathrm{GA}$). By $\sigma_A$ we denote the permutation corresponding to a superstring $A(\mathcal{S})$ constructed by $A$, and by $(l_A(1), r_A(1)), \dots, \\(l_A(n-1), r_A(n-1))$, we denote the order of merges: strings $s_{l_A(i)}$ and $s_{r_A(i)}$ are merged at step $i$. By the definition of GA we have
        \begin{gather*}
            |\overlap(s_{l_A(i)}, s_{r_A(i)})| \geq |\overlap(s_{l_A(j)}, s_{r_A(j)})|, \quad \forall\, i < j < n,
        \end{gather*}
        and if, for some $i$, $|\overlap(s_{l_A(i)}, s_{r_A(i)})| = 0$, then the same holds for any $i' > i$. We denote the first such $i$ by $T_A$ and this is the first trivial merge (that is, one with the empty overlap), after which all the merges are trivial. Note that just before step $T_A$, all the remaining strings have empty overlaps, so the resulting superstring is just a concatenation of them in some order and this order does not affect the length of the result.
        
        \begin{figure}[ht]
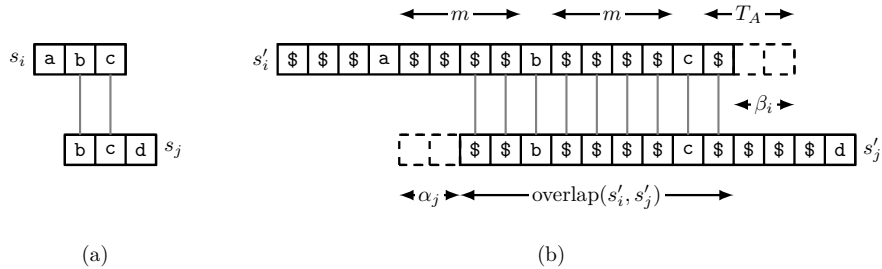

            \begin{mypic}
                \begin{scope}
                    
                    \draw (0,2) rectangle (1.5,2.5);
                    \draw[step=5mm] (0,2) grid (1.5,2.5);
                    \node[left] at (0,2.25) {$s_i$};
                    \draw (0.5,0.5) rectangle (2,1);
                    \draw[step=5mm] (0.5,0.5) grid (2,1);
                    \node[right] at (2,0.75) {$s_j$};
                    
                    \foreach \x in {0.75, 1.25}
                    \draw[gray,thick] (\x,2) -- (\x,1);
                    
                    \foreach \x/\a in {0/a, 0.5/b, 1/c}
                    \node at (\x+0.25,2.25) {\texttt{\a}};
                    \foreach \x/\a in {0.5/b, 1/c, 1.5/d}
                    \node at (\x+0.25,0.75) {\texttt{\a}};

                    \tikzmath{\shift = 4; }
                    
                    \draw (0+\shift,2) rectangle (7.5+\shift,2.5);
                    \draw[step=5mm] (0+\shift,2) grid (7.5+\shift,2.5);
                    \node[left] at (0+\shift,2.25) {$s'_i$};
                    \draw (3+\shift,0.5) rectangle (9.5+\shift,1);
                    \draw[step=5mm] (3+\shift,0.5) grid (9.5+\shift,1);
                    \node[right] at (9.5+\shift,0.75) {$s'_j$};
                    
                    \draw[dashed] (2+\shift,0.5) rectangle (3+\shift,1);
                    \draw[step=5mm, dashed] (2+\shift,0.5) grid (3+\shift,1);
                    \draw[dashed] (7.5+\shift,2) rectangle (8.5+\shift,2.5);
                    \draw[step=5mm, dashed] (7.5+\shift,2) grid (8.5+\shift,2.5);
                    
                    \foreach \x in {3.25, 3.75, ..., 7.25}
                    \draw[gray,thick] (\x+\shift,2) -- (\x+\shift,1);
                    
                    \foreach \f/\t/\y/\lab in {2/3/0/{\alpha_j}, 4.5/6.5/3/m, 2/4/3/m,
                        7/8.5/3/T_A, 3/7.5/0/{\overlap(s'_i, s'_j)}, 7.5/8.5/1.5/\beta_i}
                    \path (\f+\shift,\y) edge[<->] node[rectangle,inner sep=0.5mm,fill=white] {\strut $\lab$} (\t+\shift,\y);
                    
                    \foreach \x/\a in {0/\$, 0.5/\$, 1/\$, 1.5/a, 2/\$, 2.5/\$, 3/\$, 3.5/\$, 4/b, 4.5/\$, 5/\$, 5.5/\$, 6/\$, 6.5/c, 7/\$}
                    \node at (\x+0.25+\shift,2.25) {\texttt{\a}};
                    \foreach \x/\a in {3/\$, 3.5/\$, 4/b, 4.5/\$, 5/\$, 5.5/\$, 6/\$, 6.5/c, 7/\$, 7.5/\$, 8/\$, 8.5/\$, 9/d}
                    \node at (\x+0.25+\shift,0.75) {\texttt{\a}};
                    
                    \node at (1,-1) {(a)};
                    \node at (4.5+\shift,-1) {(b)};
                \end{scope}
            \end{mypic}
            \caption{(a)~strings $s_i$ and $s_j$ from $\mathcal{S}$. (b)~the~resulting strings $s'_i$ and $s'_j$ after disturbing; here, $m = 4$, $T_A = 3$, $\alpha_i = 1$, $\beta_i = 2$, $\alpha_j = 2$ and $\beta_j = T_A$; since $\alpha_j = \beta_i = 2$, we may conclude that $s_i$ and $s_j$ were merged by~$A$ at~step~2. }
            \label{fig:proc}
        \end{figure}

    \section{Disturbing Procedure}
    
        Here, we describe the mentioned procedure that gets rid of ties. Consider a dataset $\mathcal{S}$, an~instantiation $A \in \mathrm{GA}$ and \emph{a sentinel} $\$$ --- a symbol that does not occur in $\mathcal{S}$, and a~parameter $m$ whose value will be determined later. For every string $s_i = c_1c_2\dots c_{n_i} \in \mathcal{S}$ define a string
        \begin{gather}
            s'_i = \$^{m - \alpha_i}c_1 \$^m c_2 \$^m c_3 \$^m \dots \$^m c_{n_i} \$^{T_A - \beta_i},
        \end{gather}
        where
        \begin{enumerate}
            \item $\alpha_i$ is the number of step such that $r_{\alpha_i} = i$, if such step exists and is less than $T_A$, and $\alpha_i = T_A$ otherwise; note that if $\alpha_i < T_A$ then $s_i$ is \emph{the right} part of a non-trivial merge at step $\alpha_i$;
            \item $\beta_i$ is the number of step such that $l_{\beta_i} = i$, if such step exists and is less than $T_A$, and $\beta_i = T_A$ otherwise; note that if $\beta_i < T_A$ then $s_i$ is \emph{the left} part of a non-trivial merge at step $\beta_i$.
        \end{enumerate}
        Basically, we insert the string $\$^m$ before every character of $s_i$ and then remove some $\$$'s from the beginning of the string and add some $\$$'s to its end (see Fig.~\ref{fig:proc}). The purpose of this removal and addition is to disturb slightly overlaps of equal length, so there are no longer any ties in non-trivial merges.
        
        We denote the resulting set of disturbed strings $\{s'_1, \dots, s'_n\}$ by $\mathcal{S}'$, and all entities related to this dataset we denote by adding a prime (for example, $\sigma'_A$). Let us derive some properties of $\mathcal{S}'$.
        \begin{lemma}
            \label{lem:1}
            For all $i \neq j$, $k \neq l$
            \begin{enumerate}
                \item if $|\overlap(s_i, s_j)| = k > 0$, then $|\overlap(s'_i, s'_j)| = (m + 1)k - \alpha_j + T_A - \beta_i$;
                \item if $|\overlap(s_i, s_j)| = 0$, then $|\overlap(s'_i, s'_j)| = \min\{ T_A - \beta_i, m - \alpha_j\}$;
                \item if $m > 2n$, then disturbing preserves order on overlaps of different lengths, that is, if $|\overlap(s_i, s_j)| > |\overlap(s_k, s_l)|$, then $|\overlap(s'_i, s'_j)| > |\overlap(s'_k, s'_l)|$.
            \end{enumerate}
        \end{lemma}
        \begin{proof}
            Let $\overlap(s_i, s_j)$ be $c_1c_2\dots c_k$. Consider the string $u = \$^{m - \alpha_j}c_1 \$^m \dots \$^m c_k \$^{T_A - \beta_i}$. Clearly, $u$ is the overlap of $s'_i$ and $s'_j$ and $|u| = (m + 1)k - \alpha_j + T_A - \beta_i$. Also, that if $|\overlap(s_i, s_j)| = 0$ then $\overlap(s'_i, s'_j) = \$^{\min\{ T_A - \beta_i, m - \alpha_j\}}$.
            
            To prove the last statement, note that $\alpha_j + \beta_i \leq 2T_A < m$ and $|\overlap(s'_i, s'_j)| > (m + 1)|\overlap(s_k, s_l)| + T_A \geq |\overlap(s'_k, s'_l)|$.
        \end{proof}
    
        \begin{lemma}
            \label{lem:2}
            Let $B \in \mathrm{GA}$. Then $T_A = T'_A = T'_B$ and the first $T_A - 1$ merges are the same for both instantiations.
        \end{lemma}
        \begin{proof}
            We prove by induction that $l_A(t) = l'_A(t) = l'_B(t)$ and $r_A(t) = r'_A(t) = r'_B(t)$ for all $t < T_A$.
            
            Case $t = 1$. As $A$ is greedy, then $k_1 := |\overlap(s_{l_A(1)}, s_{r_A(1)})| \geq |\overlap(s_i, s_j)|$, for all $i \neq j$, $(i, j) \neq (l_A(1), r_A(1))$. Hence
            \begin{align*}
                |\overlap(s'_i, s'_j)| &\leq (m + 1)k_1 - \alpha_j + T_A - \beta_i < \\
                &< (m + 1)k_1 - 1 + T_A - 1 = |\overlap(s'_{l_A(1)}, s'_{r_A(1)})|,
            \end{align*}
            and $l'_A(1) = l'_B(1) = l_A(1)$ as well as $r'_A(1) = r'_B(1) = r_A(1)$.
            
            Suppose that the statement holds for all $t \leq t' < T_A - 1$. Note that at moment $t = t' + 1$ the sum $\alpha_j + \beta_i$ is strictly greater than $2t$ until $(i, j) = (l_A(t), r_A(t))$. Similarly to the base case, we have
            \begin{align*}
                |\overlap(s'_i, s'_j)| &\leq (m + 1)k_t - \alpha_j + T_A - \beta_i < \\
                &< (m + 1)k_t - t + T_A - t = |\overlap(s'_{l_A(t)}, s'_{r_A(t)})|,
            \end{align*}
            where $k_t = |\overlap(s_{l_A(t)}, s_{r_A(t)})|$, and the induction step is proven.
            
            Now note that starting from step $T_A$ all the remaining strings in $\mathcal{S}$ have empty overlaps and hence so do the remaining strings in $\mathcal{S}'$, as for all of them $\beta_i = T_A$ and the minimum in Lemma~\ref{lem:1}.2 is equal to zero. Thus, $T_A = T'_A = T'_B$ and the lemma is proven.
        \end{proof}
    
        \begin{corollary}
            \label{cor}
            As all non-trivial merges coincide, $|A(\mathcal{S}')| = |B(\mathcal{S}')|$.
        \end{corollary}

    \section{Equivalence of instantiations}
    
        \begin{theorem}
            \label{theo:one}
            If some instantiation $A$ of GA achieves $\lambda$-approximation, then so does any other instantiation.
        \end{theorem}
        \begin{proof}
            Assume the opposite and consider $B \in \mathrm{GA}$ as well as a dataset $\mathcal{S}$ such that $|B(\mathcal{S})| > \lambda|\mathrm{OPT}(\mathcal{S})|$. Let $\mathcal{S}' = \mathcal{S}'(B, m)$ be the corresponding disturbed dataset, where $m \geq n$ will be specified later.
            
            Note that $|s'_i| / m \to |s_i|$ and $|\overlap(s'_i, s'_j)| / m \to |\overlap(s_i, s_j)|$ as $m$ approaches infinity, thanks to Lemma~\ref{lem:1}.1--2. Then $|\mathrm{OPT}(\mathcal{S}')| / m \to |\mathrm{OPT}(\mathcal{S})|$, since
            \begin{align*}
                \frac 1m |\mathrm{OPT}(\mathcal{S}')| =
                &\frac 1m \min_\sigma \left\{ \sum_{i=1}^n |s'_i| - \sum_{i=1}^{n-1}|\overlap(s'_{\sigma(i)}, s'_{\sigma(i+1)})|\right\} \to \\
                \to &\min_\sigma \left\{ \sum_{i=1}^n |s_i| - \sum_{i=1}^{n-1}|\overlap(s_{\sigma(i)}, s_{\sigma(i+1)})|\right\} = |\mathrm{OPT}(\mathcal{S})|,
            \end{align*}
            $|B(\mathcal{S}')|/m \to |B(\mathcal{S})|$ and hence $|A(\mathcal{S}')|/m \to |B(\mathcal{S})|$, by Corollary~\ref{cor}.
            
            As $|B(\mathcal{S})| - \lambda|\mathrm{OPT}(\mathcal{S})| > 0$, we can choose $m$ so that $|B(\mathcal{S}')| - \lambda|\mathrm{OPT}(\mathcal{S}')|$ as well as $|A(\mathcal{S}')| - \lambda|\mathrm{OPT}(\mathcal{S}')|$ are positive. Hence $A$ is not $\lambda$-approximate.
        \end{proof}
    
        \begin{corollary}
            To prove (or disprove) the Greedy Conjecture, it is sufficient to consider datasets satisfying one of the following three properties:
            \begin{alphaenumerate}
                \item there are no ties between non-empty overlaps, that is, datasets where all the~instantiations of the~greedy algortihm work the~same;
                \item there are no empty overlaps: $\overlap(s_i, s_j) \neq \varepsilon$, $\forall\,i\neq j$;
                \item all overlaps are (pairwise) different: $|\overlap(s_i, s_j)| \neq |\overlap(s_k, s_l)|$, for all $i\neq j$, $k \neq l$, $(i, j) \neq (k, l)$.
            \end{alphaenumerate}
        \end{corollary}
        \begin{claimproof}
            \begin{alphaenumerate}
                \item Follows directly from the proof of Theorem~\ref{theo:one}, as we always can use the dataset $\mathcal{S}'$ instead of $\mathcal{S}$.
                \item Append $\$$ to each string of $\mathcal{S}'$. Then, every two strings have non-empty overlap that at least contains $\$$, and in general $T_A = T'_A = T'_B$ from Lemma~\ref{lem:2} does not hold ($T'_A$ and $T'_B$ are always $n-1$). However, the first $T_A$ merges are still the same and after them all the remaining strings have overlaps of length 1 and then the lengths of the final solutions are the same as well.
                \item Append $\$^{n(T_A - \beta_i)}$ to each string of $\mathcal{S}'$ instead of $\$^{T_A - \beta_i}$. Then
                \begin{align*}
                    |\overlap(s'_i, s'_j)| = (m + 1)k - \alpha_j + nT_A - n\beta_i,
                \end{align*}
                provided $m$ is large enough, and $\alpha_j + n\beta_i \neq \alpha_k + n\beta_l$ if $(i, j) \neq (k, l)$. Repeating the proofs of Lemmas~\ref{lem:1} and~\ref{lem:2} with this version of $\mathcal{S}'$, we obtain this statement of the lemma.
            \end{alphaenumerate}
        \end{claimproof}
    
    \section{Frequency version of SCS}
    
        Consider the following problem: given a dataset $\mathcal{S}$ and a symbol $\#$, find the superstring with the smallest number of occurrences of $\#$. We call the symbol $\#$ \emph{important}.
        
        This problem is similar to SCS: here, we also need to find the ``shortest'' superstring, but the length of the string is understood in terms of the number of occurrences of the important symbol.
        
        By $|s|_\#$ we denote the ``length'', that is, the number of occurences of $\#$ in a string $s$, and all the remaining entities related to this problem we denote by this subscript. As before, the~length$_\#$ of a superstring $s(\sigma)$ obtained from some permutation $\sigma$ is
        \begin{gather}
            |s(\sigma)|_\# = \sum_{i=1}^n |s_i|_\# - \sum_{i=1}^{n-1}|\overlap(s_{\sigma(i)}, s_{\sigma(i+1)})|_\#.
        \end{gather}
        $\mathrm{GA}_\#$ also works as GA, but instead of merging pairs of strings with longest overlaps, it merges pairs with longest$_\#$ ones.
        
        A natural question arises: is this problem easier than SCS? The following theorem claims
        that these problems are almost equivalent.
        \begin{theorem}
            \label{theo:equi}
            \begin{alphaenumerate}
                \item If there is a $\lambda$-approximate $A \in \mathrm{GA}_\#$, then $\mathrm{GA}$ is $\lambda$-approximate.
                \item If $\mathrm{GA}$ is $\lambda$-approximate, then all the instantiations of $\mathrm{GA}_\#$, that merge the longest overlaps among the longest$_\#$ ones, are $\lambda$-approximate.
            \end{alphaenumerate}
        \end{theorem}
        \begin{proof}
            \begin{alphaenumerate}
                \item Let $\mathcal{S}$ be a dataset for SCS and let $\$$ be a symbol that does not occur in the strings of $\mathcal{S}$. Transform strings of $\mathcal{S}$ by adding $\$$ before every symbol (for example, \texttt{abc} turns into $\mathtt{\$a\$b\$c}$), denote the resulting dataset by $\mathcal{S}^\$$ and consider it as a dataset for SCS$_\#$ with $\# = \$$. Clearly, $|s^\$|_\# = |s|$ and $|\overlap(s^\$, t^\$)|_\# = |\overlap(s, t)|$ for every $s \neq t$. Thus, every superstring for $\mathcal{S}$ of length $k$ corresponds to the superstring for $\mathcal{S}^\$$ of length$_\#$~$k$. Therefore, every approximate (greedy$_\#$ or not) algorithm for SCS$_\#$ induces such an algorithm (greedy$_\#$ or not, respectively) for SCS and if there is a $\lambda$-approximate instantiation of~GA$_\#$, then there is the corresponding $\lambda$-approximate instantiation of GA and by Theorem~\ref{theo:one} the entire GA is $\lambda$-approximate.
                
                \item  Assume the opposite and consider $A \in \mathrm{GA}_\#$ satisfying the property from the statement of the theorem, and a dataset $\mathcal{S}$ such that $|A(\mathcal{S})|_\# > \lambda |\mathrm{OPT}_\#(\mathcal{S})|_\#$. We construct a~new dataset $\mathcal{S}_m$ by replacing every occurence of $\#$ in the strings of $\mathcal{S}$ with~$\#^m$.
                
                Let $(l(1), r(1)), \dots, (l(n-1), r(n-1))$ be the greedy$_\#$ order or merges produced by~$A$. This order satisfies the following property: if $i < j$, then $|\overlap(s_{l(i)}, s_{r(i)})|_\# \geq |\overlap(s_{l(j)}, s_{r(j)})|_\#$, and if $|\overlap(s_{l(i)}, s_{r(i)})|_\# = |\overlap(s_{l(j)}, s_{r(j)})|_\#$, then also
                $|\overlap(s_{l(i)}, s_{r(i)})| \geq |\overlap(s_{l(j)}, s_{r(j)})|$.
                
                The key idea behind the proof is that for $m \gg n$ this order is greedy. Indeed,
                \begin{gather*}
                    |\overlap(s_m, t_m)| = |\overlap(s, t)| + (m - 1) |\overlap(s, t)|_\#,
                \end{gather*}
                so if $|\overlap(s, t)|_\# > |\overlap(u, v)|_\#$, then $|\overlap(s_m, t_m)| > |\overlap(u_m, v_m)|$ for large~$m$, and if $|\overlap(s, t)|_\# = |\overlap(u, v)|_\#$ and $|\overlap(s, t)| \geq |\overlap(u, v)|$, then $|\overlap(s_m, t_m)| \geq |\overlap(u_m, v_m)|$ for any~$m$. Hence the greedy$_\#$ order for $\mathcal{S}$ becomes the greedy order for $\mathcal{S}_m$ and the solution $A(\mathcal{S})_m$ obtained from $A(\mathcal{S})$ by replacing all the occurrences of $\#$ with $\#^m$ becomes the greedy solution for $\mathcal{S}_m$. Note that $A(\mathcal{S}_m)$ need not to be equal to $A(\mathcal{S})_m$.
                
                Since $|\mathrm{OPT}_\#(\mathcal{S}_m)| / m \to |\mathrm{OPT}_\#(\mathcal{S})|_\#$ and $|A(\mathcal{S})_m| / m \to |A(\mathcal{S})|_\#$, we may choose $m$ such that $|A(\mathcal{S})_m| - \lambda |\mathrm{OPT}_\#(\mathcal{S}_m)| > 0$. As then $|A(\mathcal{S})_m| - \lambda |\mathrm{OPT}(\mathcal{S}_m)| > 0$ as well, we obtain the greedy solution that is not $\lambda$-approximate. Then, Theorem~\ref{theo:one} implies that GA is not $\lambda$-approximate.
            \end{alphaenumerate}
        \end{proof}        
        
    \section{Further directions}
    
        The proposed connection between lengths and symbol frequencies reveals a potential frequency structure of SCS: may we think about this problem in terms of frequencies rather than lengths? We know that there is a correspondence between greedy and greedy$_\#$ superstrings, but are the greedy solutions 2-approximate$_\#$ and vice versa? May we treat the greedy solutions not only as reasonably short ones, but also as solutions that contain reasonably small number of~occurrences of each symbol uniformly?
        
    \section*{Acknowledgments}
    
        Many thanks to Alexander Kulikov for valuable discussions and proofreading the text.


\begin{thebibliography}{10}
        
        \bibitem{BJLTY1991}
        Avrim Blum, Tao Jiang, Ming Li, John Tromp, and Mihalis Yannakakis.
        \newblock {Linear approximation of shortest superstrings}.
        \newblock In {\em STOC 1991}, pages 328--336. ACM, 1991.
        
        \bibitem{phdthesis}
        John Gallant.
        \newblock {\em String compression algorithms.}
        \newblock PhD thesis, Princeton, 1982.
        
        \bibitem{GMS1980}
        John Gallant, David Maier, and James~A. Storer.
        \newblock {On finding minimal length superstrings}.
        \newblock {\em J. Comput. Syst. Sci.}, 20(1):50--58, 1980.
        
        \bibitem{gevezes2014shortest}
        Theodoros~P. Gevezes and Leonidas~S. Pitsoulis.
        \newblock {\em The shortest superstring problem}, pages 189--227.
        \newblock Springer, 2014.
        
        \bibitem{GMKLN2019}
        Alexander Golovnev, Alexander~S Kulikov, Alexander Logunov, Ivan Mihajlin, and
        Maksim Nikolaev.
        \newblock Collapsing superstring conjecture.
        \newblock In {\em Approximation, Randomization, and Combinatorial Optimization.
            Algorithms and Techniques (APPROX/RANDOM 2019)}. Schloss
        Dagstuhl-Leibniz-Zentrum fuer Informatik, 2019.
        
        \bibitem{KS2005}
        Haim Kaplan and Nira Shafrir.
        \newblock {The greedy algorithm for shortest superstrings}.
        \newblock {\em Inf. Process. Lett.}, 93(1):13--17, 2005.
        
        \bibitem{mucha2007tutorial}
        Marcin Mucha.
        \newblock A tutorial on shortest superstring approximation.
        \newblock \url{https://www.mimuw.edu.pl/~mucha/teaching/aa2008/ss.pdf}.
        \newblock Online; accessed 10 February 2021.
        
        \bibitem{pevzner2001eulerian}
        Pavel~A. Pevzner, Haixu Tang, and Michael~S. Waterman.
        \newblock An eulerian path approach to {DNA} fragment assembly.
        \newblock {\em Proc. Natl. Acad. Sci. U.S.A.}, 98(17):9748--9753, 2001.
        
        \bibitem{storer1987data}
        James~A. Storer.
        \newblock {\em Data compression: methods and theory}.
        \newblock Computer Science Press, Inc., 1987.
        
        \bibitem{TU1988}
        Jorma Tarhio and Esko Ukkonen.
        \newblock {A greedy approximation algorithm for constructing shortest common
            superstrings}.
        \newblock {\em Theor. Comput. Sci.}, 57(1):131--145, 1988.
        
        \bibitem{T1989}
        Jonathan~S. Turner.
        \newblock {Approximation algorithms for the shortest common superstring
            problem}.
        \newblock {\em Inf. Comput.}, 83(1):1--20, 1989.
        
        \bibitem{waterman1995introduction}
        Michael~S. Waterman.
        \newblock {\em Introduction to computational biology: maps, sequences and
            genomes}.
        \newblock CRC Press, 1995.
        
    \end{thebibliography}
    \end{document}